\documentclass[11pt]{amsart}
\usepackage[margin=0.85in]{geometry}
\ifdefined\pdfglyphtounicode
\pdfglyphtounicode{parenleftbig}{0028}
\pdfglyphtounicode{parenrightbig}{0029}
\pdfglyphtounicode{parenleftBig}{0028}
\pdfglyphtounicode{parenrightBig}{0029}
\pdfglyphtounicode{parenleftbigg}{0028}
\pdfglyphtounicode{parenrightbigg}{0029}
\fi
\usepackage{amsmath,amssymb,amsthm,mathtools}
\usepackage{booktabs,array}
\usepackage{microtype,enumitem,aliascnt}
\usepackage[hidelinks]{hyperref}
\usepackage[nameinlink,noabbrev]{cleveref}
\hypersetup{pdftitle={The A3 Narain family over Z4},
 pdfauthor={Abdellah Sebbar and Rajae Sammani}}
\newtheorem{theorem}{Theorem}[section]
\newaliascnt{proposition}{theorem}
\newtheorem{proposition}[proposition]{Proposition}
\aliascntresetthe{proposition}
\newaliascnt{lemma}{theorem}
\newtheorem{lemma}[lemma]{Lemma}
\aliascntresetthe{lemma}
\newaliascnt{corollary}{theorem}
\newtheorem{corollary}[corollary]{Corollary}
\aliascntresetthe{corollary}
\theoremstyle{remark}
\newaliascnt{remark}{theorem}
\newtheorem{remark}[remark]{Remark}
\aliascntresetthe{remark}
\crefname{theorem}{Theorem}{Theorems}
\crefname{proposition}{Proposition}{Propositions}
\crefname{lemma}{Lemma}{Lemmas}
\crefname{corollary}{Corollary}{Corollaries}
\crefname{remark}{Remark}{Remarks}
\newcommand{\Z}{\mathbb Z}
\newcommand{\R}{\mathbb R}

\newcommand{\F}{\mathbb F}
\newcommand{\C}{\mathcal C}
\newcommand{\Lam}{\Lambda}
\newcommand{\HH}{\mathcal H}
\newcommand{\diag}{\operatorname{diag}}
\newcommand{\cA}{\mathsf A}

\begin{document}

\title{The \texorpdfstring{$A_3$}{A3} Narain family over \texorpdfstring{$\Z_4$}{Z4}}
\author{Rajae Sammani} \author{Abdellah Sebbar}


\begin{abstract}
The code $\{000,220,202,022\}\subset\Z_4^3$ realizes the split $A_3$ root-weight gluing construction as ordinary Construction A, preserving the Lorentzian pairing and the positive energy form. For two hyperbolic pairs of $A_3$ blocks, the twenty-two even self-dual gluing codes form seven block-monomial classes. Explicit metric and $B$-field representatives and theta series distinguish exactly seven T-duality classes in this family at central charge six; the one-pair family has two classes at central charge three. For arbitrary $d$, the chiral root system on either side is $A_1^{3m}$, where $m$ counts the words $(2e_i,2e_i)$ in the gluing code. Exactly one two-pair class has no zero-$B$ T-duality frame. Intersections of gluing codes determine finite shift orbifolds between the backgrounds, with the target code specifying the twisted-sector projections.

\end{abstract}

\maketitle
\section{Introduction}

Code constructions provide finite descriptions of selected Narain conformal field theories \cite{AlamEtAl,DymarskyShapere}. In the root-weight construction, a code specifies the cosets of a root lattice that occur in the momentum lattice \cite{MizoguchiOikawa,SaidiSammani}. We consider the $A_3$ case, with its standard Euclidean form and discriminant group $A_3^*/A_3\simeq\Z_4$, where $\Z_4=\Z/4\Z$.

The code
\begin{equation}\label{eq:Edef}
 E=\{000,220,202,022\}\subset\Z_4^3
\end{equation}
is twice the binary even-weight code of length three. In the standard Euclidean normalization,
\[
 \tfrac12\{z\in\Z^3:z\bmod4\in E\}=D_3\simeq A_3,
 \qquad E^\perp/E=\langle g+E\rangle\simeq\Z_4,
 \quad g=(1,1,1).
\]
For an even self-dual gluing code $\C\subset\Z_4^{2d}$, define
\[
 \widehat\C=\bigcup_{(a_1,\ldots,a_{2d})\in\C}
 (a_1g+E)\times\cdots\times(a_{2d}g+E)\subset\Z_4^{6d}.
\]
Theorem~\ref{thm:main} identifies its ordinary Construction A lattice with the $A_3$ momentum lattice, preserving both forms.

For $d=1$ there are three labeled codes and two block-monomial classes.
For $d=2$, reduction modulo two gives
\[
\begin{array}{c|c|c}
\text{binary data}&\text{code type}&\text{number of codes}\\\hline
6\text{ totally singular planes}&\Z_4^2&12\\
9\text{ nonzero singular vectors}&\Z_4\oplus\Z_2^2&9\\
0&\Z_2^4&1
\end{array}
\]
The twenty-two codes form seven block-monomial orbits of sizes
$4,4,4,4,4,1,1$. Integral bases determine their Narain data, and distinct
theta series distinguish their T-duality classes.

The split pairing differs from the left-minus-right discriminant form
used for level-one $\mathrm{SU}(4)$ modular invariants in \cite{Angelinos};
see Remark~\ref{rem:WZW}. The two forms define different code
classification problems over $\Z_4$.

For arbitrary $d$, the chiral roots form orthogonal $A_1$ factors.
Among the two-pair classes, only $\mathrm M_2$ has no zero-$B$ T-duality
frame; the obstruction is the irreducible spanning $D_{12}$ root system
of its positive momentum lattice. The shift orbifold
$\mathrm T\to\mathrm M_2$ gauges $(-1)^{2\Delta}$, with the twisted
coset specified by $\mathrm M_2$.

\section{The fixed lattice and the gluing codes}\label{sec:fixed}

\subsection{The discriminant group of \texorpdfstring{$A_3$}{A3}}

We use the root lattice
\[
 R=A_3=\{x\in\Z^4:x_1+x_2+x_3+x_4=0\}
\]
with the usual Euclidean inner product. It has rank three, with simple-root
basis $\alpha_i=\varepsilon_i-\varepsilon_{i+1}$, $1\le i\le3$, where the
$\varepsilon_i$ are the standard coordinate vectors of $\R^4$. Roots have
squared norm two; this fixes the Euclidean normalization throughout.
The Gram matrix $K=(\alpha_i\cdot\alpha_j)$ and its inverse are
\begin{equation}\label{eq:K}
 K=\begin{pmatrix}2&-1&0\\-1&2&-1\\0&-1&2\end{pmatrix},
 \qquad K^{-1}=\frac14\begin{pmatrix}3&2&1\\2&4&2\\1&2&3\end{pmatrix}.
\end{equation}
The weight lattice is the Euclidean dual
\[
 W=R^*=\{w\in\operatorname{span}_{\R}R:w\cdot r\in\Z
                 \text{ for every }r\in R\}.
\]
In simple-root coordinates, $R=\Z^3$ and $W=K^{-1}\Z^3$; hence both have
rank three and $|W/R|=\det K=4$. The fundamental weight
$\omega_3=(1,1,1,-3)/4$ has squared norm $3/4$, with $4\omega_3\in R$
and $2\omega_3\notin R$. Thus $W/R=\langle\omega_3+R\rangle\simeq\Z_4$.

We use the three-coordinate realization \cite{ConwaySloane}
\[
 D_3=\{x\in\Z^3:x_1+x_2+x_3\equiv0\pmod2\}.
\]
This is a rank-three sublattice of index two in $\Z^3$. The vectors
\[
 v_1=(0,1,-1),\qquad v_2=(1,-1,0),\qquad v_3=(0,1,1)
\]
have Gram matrix $K$ and column determinant of absolute value two.
They generate a sublattice of $D_3$ of the same index in $\Z^3$, and
therefore form a basis of $D_3$. The map $\alpha_i\mapsto v_i$ is an
isometry $R\simeq D_3$; it identifies $W$ with $D_3^*$ and sends
$\omega_3$ to $(1,1,1)/2$.

Recall $E=\{000,220,202,022\}\subset\Z_4^3$ from \eqref{eq:Edef}.
It is obtained from the binary even-weight code $\{000,110,101,011\}$
by the additive embedding $b\mapsto2b$ from $\F_2$ to $\Z_4$.
Normalized ordinary Construction A over $\Z/M\Z$ takes the inverse image
of a code in $\Z^N$ and multiplies it by $1/\sqrt M$. Here $M=4$, so
the factor is $1/2$. For $E$, orthogonality is taken with respect to
the ordinary dot product modulo four. Since all its entries are even,
$E\subset E^\perp$.

\begin{lemma}\label{lem:inner}
For the code $E$ in \eqref{eq:Edef},
\begin{align}
 \frac12\{n\in\Z^3:n\bmod4\in E\}&=D_3,\\
 E^\perp&=\{z\in\Z_4^3:z_1\equiv z_2\equiv z_3\pmod2\}.
\end{align}
Moreover
\[
 \frac12\{n\in\Z^3:n\bmod4\in E^\perp\}=D_3^*,
 \qquad E^\perp/E=\langle g+E\rangle\simeq\Z_4,
 \quad g=(1,1,1).
\]
The dot product on the finite quotient is $3ab\pmod4$ on
$(ag+E,bg+E)$. The corresponding discriminant pairing on $D_3^*/D_3$
is $3ab/4\pmod\Z$.
\end{lemma}
\begin{proof}
The words of $E$ are precisely the even triples whose coordinate sum
is divisible by four. Thus $n\bmod4\in E$ if and only if $n=2z$ with
$z\in D_3$, proving the first identity. Orthogonality to the generators
$(2,2,0)$ and $(2,0,2)$ gives
\[
 2(z_1+z_2)=2(z_1+z_3)=0\pmod4,
\]
so the three entries of a word in $E^\perp$ have the same parity.
Its integer inverse image, divided by two, is consequently
\[
 \Z^3\ \cup\ (\Z+\tfrac12)^3.
\]
To identify this union with $D_3^*$, let $e_i$ be the standard coordinate
vectors of $\R^3$. If $x$ pairs integrally with $D_3$, pairing with
$e_i-e_j$ and $e_i+e_j$ gives $x_i-x_j\in\Z$ and $2x_i\in\Z$.
Hence its coordinates are all integral or all half-integral. Conversely,
both kinds of vectors pair integrally with $D_3$, since its vectors
have even coordinate sum.

There are eight all-even and eight all-odd words in $E^\perp$, so
$|E^\perp/E|=16/4=4$. Since $g\in E^\perp$, $2g=(2,2,2)\notin E$
and $4g=0$, its class generates the quotient. For $e,f\in E$,
\[
 (ag+e)\cdot(bg+f)\equiv ab\,g\cdot g=3ab\pmod4.
\]
This also shows that the quotient pairing is independent of the
representatives. The map $n/2\mapsto(n\bmod4)+E$ identifies
$D_3^*/D_3$ with $E^\perp/E$. Dividing numerator vectors by two
divides their inner product by four, giving the discriminant pairing.
\end{proof}

We identify the label $a\in\Z_4$ with $ag+E$ and write
\begin{equation}\label{eq:cosets}
 W_a=\left\{\frac n2:n\in\Z^3,\ n\bmod4\in ag+E\right\}\subset D_3^*.
\end{equation}
Equivalently, $W_a=ag/2+D_3$. These are the four cosets of $R$ in $W$
under the fixed isometry, and $[x]=a$ means $x\in W_a$. In particular,
$W_0=D_3$ and $W_0\cup W_2=\Z^3$. A useful formula for the label is
\begin{equation}\label{eq:blocklabel}
 n\bmod4\in ag+E\quad\Longrightarrow\quad
 a\equiv3(n_1+n_2+n_3)\pmod4.
\end{equation}
Indeed the coordinate sum of a word of $E$ is zero modulo four, so
that of $ag+e$ is $3a$; multiplication by three recovers $a$.

\subsection{The split finite pairing}

For $d\ge1$, put $M_d=\Z_4^d\oplus\Z_4^d$ and write its elements as
$(a,b)$, where $a=(a_1,\ldots,a_d)$ and $b=(b_1,\ldots,b_d)$.
Use the forms
\begin{equation}\label{eq:splitfinite}
 \beta((a,b),(a',b'))=a\cdot b'+b\cdot a'\pmod4,
 \qquad q_d(a,b)=a\cdot b\pmod4.
\end{equation}
Here $a\cdot b=\sum_i a_ib_i$. The polarization identity is
\[
 q_d(u+v)-q_d(u)-q_d(v)=\beta(u,v).
\]
A \emph{gluing code} is an additive subgroup $\C\subset M_d$, with
orthogonal complement
\[
 \C^\perp=\{z\in M_d:\beta(c,z)=0\text{ for every }c\in\C\}.
\]
This uses the split pairing $\beta$, unlike the ordinary dot product
used for $E$. We call $\C$ \emph{even self-dual} when
\begin{equation}\label{eq:evenfinite}
 q_d|_\C=0,\qquad \C=\C^\perp.
\end{equation}
Evenness implies $\C\subseteq\C^\perp$ by polarization. It is stronger
than vanishing of self-pairings, since $\beta(c,c)=2q_d(c)$ and two
is not invertible in $\Z_4$.

The pairing is nondegenerate: if $u_i$ denotes the $i$th coordinate
vector of $\Z_4^d$, then
\[
 \beta((a,b),(u_i,0))=b_i,\qquad
 \beta((a,b),(0,u_i))=a_i.
\]
Thus pairing with coordinate vectors recovers every coordinate.
Character orthogonality gives
\[
 |\C|\,|\C^\perp|
 =\sum_{z\in M_d}\sum_{c\in\C}e^{2\pi i\beta(c,z)/4}
 =4^{2d}.
\]
For fixed $z$, the inner sum is $|\C|$ if $z\in\C^\perp$ and zero
otherwise, being the sum of a nontrivial character of $\C$.
Reversing the sums, nondegeneracy makes every term except $c=0$ vanish.
Consequently a self-dual code has order $4^d$; conversely, an even
code of that order is self-dual.

The coordinates $a_i,b_i$ form a hyperbolic pair: their coordinate
vectors have self-pairing zero and mutual pairing one. The code
equivalence used below is generated by permutations of these $d$ pairs,
independent exchanges $a_i\leftrightarrow b_i$, and simultaneous sign
changes $(a_i,b_i)\mapsto(-a_i,-b_i)$ within any chosen pair.
Each operation preserves $\beta$ and $q_d$. We call the resulting
equivalence \emph{block-monomial equivalence}; a labeled code is a
subgroup before taking this equivalence relation.

\begin{remark}\label{rem:WZW}
The finite pairing for level-one $\mathrm{SU}(4)$ modular invariants is
\[
 3(a\cdot a'-b\cdot b')\pmod4,
\]
rather than the split pairing \eqref{eq:splitfinite}; see
\cite[Sections~2.2 and~3.4]{Angelinos}. These forms are not isometric over
$\Z_4$. Every self-pairing for the split form is even, whereas a unit
coordinate vector has odd self-pairing for the difference form. This
property is preserved by every invertible change of coordinates.
Thus the $N=4$, $n=2$ enumeration in that reference and the split-code
classification below concern different finite quadratic modules.
\end{remark}

\subsection{The Narain background}

For an even self-dual gluing code $\C\subset M_d$, define
\begin{equation}\label{eq:rootlift}
 \Lam_{A_3}(\C)=\{(x,y)\in W^d\oplus W^d:
 ([x_1],\ldots,[x_d],[y_1],\ldots,[y_d])\in\C\}.
\end{equation}
Each $x_i,y_i$ is a vector in the rank-three lattice $W$, while its
bracket is a label in $W/R\simeq\Z_4$. If $\pi$ denotes this label map,
then $\Lam_{A_3}(\C)=\pi^{-1}(\C)$ and
\[
 R^d\oplus R^d\subseteq\Lam_{A_3}(\C)\subseteq W^d\oplus W^d,
 \qquad \Lam_{A_3}(\C)/(R^d\oplus R^d)\simeq\C.
\]
Thus $\C$ is a finite subgroup specifying the allowed cosets, whereas
$\Lam_{A_3}(\C)$ is a lattice of rank $6d$. For example, when $d=1$,
a codeword $(1,0)$ contributes the coset $W_1\times W_0$.
In numerator coordinates its residues form $(g+E)\times E$, a set of
$4\cdot4=16$ words of length six. This is the symbol replacement used
in Section~\ref{sec:ordinary}.

On the ambient real space $\R^{3d}\oplus\R^{3d}$ use
\begin{equation}\label{eq:forms}
 \eta((x,y),(x',y'))=x\cdot y'+y\cdot x',
 \qquad \HH((x,y),(x',y'))=x\cdot x'+y\cdot y'.
\end{equation}
The dot products sum the fixed Euclidean products in the $A_3$ blocks.
The form $\eta$ has signature $(3d,3d)$, while $\HH$ is positive definite.
For vectors $z,z'$ with labels $c,c'$, Lemma~\ref{lem:inner} gives
\[
 \eta(z,z')\equiv\tfrac34\beta(c,c')\pmod\Z,
 \qquad \tfrac12\eta(z,z)\equiv\tfrac34q_d(c)\pmod\Z.
\]
Hence $\Lam_{A_3}(\C)$ is integral and even for $\eta$.
Its $\eta$-dual lies in $W^d\oplus W^d$ because it must pair integrally
with $R^d\oplus R^d$; the first congruence then shows that this dual is
$\pi^{-1}(\C^\perp)=\Lam_{A_3}(\C)$. It is therefore an even unimodular
lattice for $\eta$.

Put $n=3d$. The left and right momenta are
\[
 p_L=\frac{x+y}{\sqrt2},\qquad p_R=\frac{x-y}{\sqrt2}.
\]
Allowed momentum pairs are those obtained from
$(x,y)\in\Lam_{A_3}(\C)$. The associated Narain theory is a conformal
field theory of $n$ compact free bosons, with $c=\bar c=n$
\cite{AlamEtAl,DymarskyShapere}. For a momentum state without oscillator
excitations, its conformal weights, scaling dimension, and spin are
\[
 h=\frac{\|p_L\|^2}{2},\qquad \bar h=\frac{\|p_R\|^2}{2},
 \qquad \Delta=h+\bar h=\frac{\HH(z,z)}2,
 \qquad s=h-\bar h=\frac{\eta(z,z)}2.
\]
Thus $\eta$ determines the spin, which is integral by evenness, and
$\HH$ determines the scaling dimension. On a circle of circumference
$L$, in units with $\hbar$ and propagation speed one, the energy above
the vacuum is $2\pi\Delta/L$. This explains the term \emph{positive
energy form}: scaling dimension equals energy above the vacuum only
after fixing $L=2\pi$. Oscillator levels $N_L,N_R\ge0$ add $N_L,N_R$
to the two conformal weights.

A \emph{Narain background} here consists of the lattice together with
$\eta$ and $\HH$. An equivalence of backgrounds is a lattice isomorphism
preserving both forms. Since $\eta$ and $\HH$ become the difference and
sum of the left and right Euclidean pairings, this is equivalent to
identifying the momentum lattices by independent orthogonal
transformations of $p_L$ and $p_R$.

In an integral null basis, whose vectors $u_i,v_i$ satisfy
$\eta(u_i,u_j)=\eta(v_i,v_j)=0$ and $\eta(u_i,v_j)=\delta_{ij}$,
the integral duality group is
\[
 O(n,n;\Z)=\{U\in GL_{2n}(\Z):U^T\eta_0U=\eta_0\},
 \qquad \eta_0=\begin{pmatrix}0&I_n\\I_n&0\end{pmatrix}.
\]
Here integrality refers to coordinates in the lattice basis.
If $\mathsf H_1,\mathsf H_2$ are the positive Gram matrices of two
backgrounds in such bases, equivalence also requires
$U^T\mathsf H_2U=\mathsf H_1$. Explicit bases are given in
Section~\ref{sec:GB}. Every block-monomial code operation induces a
background equivalence, but general Narain equivalences may mix the
$A_3$ blocks. The spectral comparisons below decide whether distinct
code classes admit further identifications.

\section{An explicit ordinary \texorpdfstring{$\Z_4$}{Z4} realization}\label{sec:ordinary}

Set $E_d=E^{\oplus2d}$ and define
\begin{equation}\label{eq:rho}
 \rho:(E^\perp)^{2d}\longrightarrow\Z_4^{2d},
 \qquad (a_1g+e_1,\ldots,a_{2d}g+e_{2d})\longmapsto(a_1,\ldots,a_{2d}).
\end{equation}
By Lemma~\ref{lem:inner}, this is a surjective homomorphism with kernel $E^{\oplus2d}$. Each block has a unique discriminant label. For a gluing code $\C$, put
\begin{equation}\label{eq:Chat}
 \widehat\C=\rho^{-1}(\C)
 =\{(a_1g+e_1,\ldots,a_{2d}g+e_{2d}):a\in\C,\ e_i\in E\}.
\end{equation}
Writing a word as $(u_1,\ldots,u_d,v_1,\ldots,v_d)$ with
$u_i,v_i\in\Z_4^3$, the ordinary split pairing is
$\sum_i(u_i\cdot v'_i+v_i\cdot u'_i)\pmod4$. Thus the $i$th
length-three block in the first half is paired with the $i$th block
in the second half.

\begin{theorem}\label{thm:main}
Let $\C\subset\Z_4^{2d}$ be even self-dual for \eqref{eq:splitfinite}. Then $\widehat\C\subset\Z_4^{6d}$ is even self-dual for the ordinary split pairing, and
\begin{equation}\label{eq:latticeiso}
 \Lam_{A_3}(\C)\simeq\Lam_{\cA}(\widehat\C),
 \qquad
 \Lam_{\cA}(\widehat\C)=\frac12\{z\in\Z^{6d}:z\bmod4\in\widehat\C\}.
\end{equation}
After the fixed blockwise isometry $A_3\simeq D_3$, the two sides are the same subset of the real momentum space. This identification preserves both $\eta$ and $\HH$. In particular, the two constructions give the same Narain CFT.
\end{theorem}
\begin{proof}
Write a word of $(E^\perp)^{2d}$ as $z=(u_1,\ldots,u_d,v_1,\ldots,v_d)$, where $u_i=a_i g+e_i$ and $v_i=b_i g+f_i$. Lemma~\ref{lem:inner} gives, for two such words with labels $(a,b)$ and $(a',b')$,
\begin{align*}
 \sum_i(u_i\cdot v_i'+v_i\cdot u_i')
   &\equiv3\beta((a,b),(a',b'))\pmod4,\\
 \sum_i u_i\cdot v_i&\equiv3q_d(a,b)\pmod4.
\end{align*}
Since three is a unit modulo four, the evenness and isotropy of $\C$ imply the corresponding conditions for $\widehat\C$. The kernel of $\rho$ is $E^{\oplus2d}$, so
\[
 |\widehat\C|=|E|^{2d}|\C|=4^{2d}4^d=4^{3d}.
\]
The ordinary split pairing on $\Z_4^{6d}$ is nondegenerate. Consequently $|\widehat\C^\perp|=4^{6d}/|\widehat\C|=4^{3d}$. Isotropy gives $\widehat\C\subset\widehat\C^\perp$, and equality of orders proves self-duality.

Write a vector of $(D_3^*)^{2d}$ as $z/2$, where $z\in\Z^{6d}$ and each length-three block reduces into $E^\perp$. Its labels modulo $D_3^{2d}$ are $\rho(z\bmod4)$. It belongs to \eqref{eq:rootlift} exactly when these labels lie in $\C$, or equivalently when $z\bmod4\in\widehat\C$. The momentum sets in \eqref{eq:latticeiso} are therefore equal in the $D_3$ coordinates. The fixed blockwise Euclidean isometry preserves both sums in \eqref{eq:forms}, and hence the left and right squared momenta separately.

For the lattice conditions, write the ordinary coordinates as $(U,V)/2$, with $U,V\in\Z^{3d}$. The finite conditions just proved imply
\[
 \eta\bigl((U,V)/2,(U',V')/2\bigr)
   =\frac{U\cdot V'+V\cdot U'}4\in\Z,
 \qquad
 \eta\bigl((U,V)/2,(U,V)/2\bigr)
   =\frac{U\cdot V}{2}\in2\Z.
\]
Thus the ordinary lattice is integral and even. It contains $2\Z^{6d}$ with index $|\widehat\C|=4^{3d}$. The absolute determinant of $2\Z^{6d}$ for $\eta$ is $4^{6d}$, so the determinant of the ordinary lattice has absolute value $4^{6d}/(4^{3d})^2=1$. It is therefore even unimodular, with the fixed positive form used in Narain Construction A \cite[Section~4.1]{AlamEtAl}.
\end{proof}

\begin{corollary}\label{cor:nonew}
The ordinary codes arising from this $A_3$ construction are exactly the even self-dual split codes $\mathcal B$ with
\[
 E^{\oplus2d}\subset\mathcal B\subset(E^\perp)^{\oplus2d}.
\]
Their gluing codes are recovered by $\C=\rho(\mathcal B)$.
\end{corollary}
\begin{proof}
If $\mathcal B$ contains $\ker\rho=E^{\oplus2d}$, the subgroup correspondence gives $\mathcal B=\rho^{-1}(\rho(\mathcal B))$. Put $\C=\rho(\mathcal B)$. The two pairing congruences in the proof of Theorem~\ref{thm:main} show that $q_d$ vanishes on $\C$ and that $\C\subset\C^\perp$, because multiplication by three is invertible modulo four. Moreover self-duality of $\mathcal B$ gives
\[
 |\C|=\frac{|\mathcal B|}{|E|^{2d}}
      =\frac{4^{3d}}{4^{2d}}=4^d.
\]
Nondegeneracy of $\beta$ implies $|\C^\perp|=4^{2d}/|\C|=4^d$, so $\C=\C^\perp$. Conversely Theorem~\ref{thm:main} constructs such an ordinary code from every even self-dual $\C$, and its inverse image contains the kernel and lies in the domain of $\rho$.
\end{proof}

For each generator $v=(a_1,\ldots,a_{2d})$ of $\C$, take its lift $(a_1g,\ldots,a_{2d}g)$ and adjoin the two generators $(2,2,0)$ and $(2,0,2)$ in each block. These vectors generate $\widehat\C$, since their images generate $\C$ and the block generators span $\ker\rho$.

\section{The gluing codes for one and two pairs}\label{sec:classification}

The classification depends only on $(M_d,\beta,q_d)$; the dependence
of the resulting backgrounds on the base lattice is discussed in
Remark~\ref{rem:base}.

\subsection{One hyperbolic pair}
There are exactly three even self-dual subgroups of $\Z_4^2$:
\[
 \C_x=\{(a,0):a\in\Z_4\},\qquad
 \C_y=\{(0,b):b\in\Z_4\},\qquad
 \C_2=2\Z_4\oplus2\Z_4.
\]
Indeed each has order four. A cyclic subgroup of order four has a generator with an odd entry. The condition $ab=0$ then forces the other entry to vanish, giving $\C_x$ or $\C_y$. A noncyclic subgroup of order four is the full two-torsion subgroup, which is even self-dual. Exchanging the two coordinates identifies $\C_x$ with $\C_y$, leaving two block-monomial classes, denoted $\mathrm F$ and $\mathrm T$.

\subsection{Two hyperbolic pairs}\label{sec:d2class}

For two pairs, order the coordinates of $\mathbb Z_4^4$ as
\[
(a_1,a_2,b_1,b_2).
\]
Set $J=\left(\begin{smallmatrix}0&1\\-1&0\end{smallmatrix}\right)$.

Let
\[
V=M_2/2M_2\simeq\mathbb F_2^4,
\qquad
\bar q(a,b)=a\cdot b,
\]
with polar form
\[
\bar\beta((a,b),(a',b'))=a\cdot b'+b\cdot a'.
\]
This is the four-dimensional split quadratic space over $\mathbb F_2$.

\begin{lemma}\label{lem:F2geometry}
The quadratic space $(V,\bar q)$ has nine nonzero singular vectors and six maximal totally singular planes.  Under the reduction modulo $2$ of the block-monomial group:
\begin{enumerate}[label=\textup{(\roman*)}]
\item the six singular planes split into an orbit of four coordinate planes and an orbit of two diagonal planes;
\item the nine nonzero singular vectors split into three orbits, represented by
\[
0001,\qquad0011,\qquad1111,
\]
of sizes $4,4,1$, respectively.
\end{enumerate}
\end{lemma}

\begin{proof}
Identify $(a_1,a_2,b_1,b_2)$ with
\[
 \begin{pmatrix}a_1&a_2\\b_2&b_1\end{pmatrix}\in M_2(\F_2).
\]
Its determinant is $\bar q(a,b)$. A nonzero singular matrix has rank one, so it has the form $uv^T$ with $u,v\in\F_2^2\setminus\{0\}$. This expression is unique, since the only nonzero scalar in $\F_2$ is one. There are therefore $3\cdot3=9$ nonzero singular vectors.

Consider independent rank-one matrices $uv^T$ and $xy^T$. If $u,x$ and $v,y$ are both independent, then
\[
 uv^T+xy^T=(u\ \ x)\begin{pmatrix}v^T\\y^T\end{pmatrix}
\]
is invertible. Otherwise their sum is nonzero of rank one and has the same column line or the same row line as both matrices. Thus a two-dimensional totally singular subspace consists either of the matrices with a fixed column line or of those with a fixed row line. Each choice gives a plane; there are three of each. These planes are maximal: a totally singular subspace $U$ satisfies $U\subset U^{\perp_{\bar\beta}}$, and nondegeneracy of $\bar\beta$ gives $\dim U\le2$.

Four of these planes are obtained by choosing one coordinate line from each hyperbolic pair. The remaining two are
\[
 \langle1001,0110\rangle,\qquad\langle1100,0011\rangle.
\]
Independent exchanges in the two pairs are transitive on the four coordinate planes. Exchanging $a_2$ and $b_2$ interchanges the two displayed planes. The two orbits cannot merge, since a coordinate plane contains a vector with exactly one nonzero coordinate, whereas a diagonal plane does not.

For a singular vector, the two products $a_1b_1$ and $a_2b_2$ must agree. If both vanish, each pair has at most one nonzero entry. There are four choices with only one pair nonzero and four with both pairs nonzero. The exchanges and the pair permutation act transitively on each set. If both products equal one, the vector is $1111$, which is fixed. This gives the asserted vector orbits and representatives.
\end{proof}

A subgroup of $\mathbb Z_4^4$ has exponent dividing four. At order sixteen, its possible abstract group types are
\[
\mathbb Z_4^2,
\qquad
\mathbb Z_4\oplus\mathbb Z_2^2,
\qquad
\mathbb Z_2^4.
\]
The reduction modulo two distinguishes these three types.

\begin{proposition}\label{prop:d2count}
There are exactly
\[
12\text{ codes of type }\mathbb Z_4^2,
\qquad
9\text{ codes of type }\mathbb Z_4\oplus\mathbb Z_2^2,
\qquad
1\text{ code of type }\mathbb Z_2^4.
\]
\end{proposition}

\begin{proof}
The kernel of reduction modulo two on $\C$ is $\C[2]=\C\cap2\Z_4^4$. For the three possible abstract types, $\C[2]$ has orders $4,8,16$, respectively. Since $|\C|=16$, their images in $V$ have dimensions $2,1,0$. The image is totally singular: $q_2$ vanishes on $\C$, and reducing $q_2(c+c')-q_2(c)-q_2(c')=\beta(c,c')$ modulo two gives the polar condition.

Suppose first that $\C\simeq\Z_4^2$. Its image $W$ is one of the six singular planes in Lemma~\ref{lem:F2geometry}. Over a coordinate plane we may arrange $W=\{(a,0):a\in\F_2^2\}$. Choose two generators of $\C$. Their projected $a$-coordinate columns are independent modulo two, so the determinant of their $2\times2$ matrix is odd. That matrix is invertible over $\Z_4$. The projection $\C\to\Z_4^2$ is consequently an isomorphism, and there is a unique matrix $A$ with
\[
 \C=\{(a,Aa):a\in\Z_4^2\},\qquad A\equiv0\pmod2.
\]
Pairing two graph vectors shows that isotropy is $A+A^T=0$. Applying $a^TAa=0$ to the two coordinate vectors forces $A_{11}=A_{22}=0$. Conversely these diagonal conditions and $A+A^T=0$ make $a^TAa$ vanish for every $a$. Hence $A=bJ$ with $b=0$ or $2$. Each graph has order sixteen and is even isotropic, hence self-dual. Distinct $b$ give distinct graphs. Thus each of the four coordinate planes has exactly two lifts.

One diagonal plane is $\operatorname{graph}(J\bmod2)$; the other is its image under the exchange $a_2\leftrightarrow b_2$. For the graph plane the same invertible-projection argument applies, now with $A\equiv J\pmod2$. The even isotropic graphs are exactly $A=bJ$ with $b=1,3$. The coordinate exchange transports these two lifts bijectively to the other diagonal plane. The number of free codes is therefore $4\cdot2+2\cdot2=12$.

Now suppose $\C\simeq\Z_4\oplus\Z_2^2$. Its image is generated by one nonzero singular vector $r$. An order-four word $v\in\C$ reduces to $r$, and every element of $\C[2]$ is $2t$ for a unique $t\in V$. The congruence
\[
 \beta(v,2t)\equiv2\bar\beta(r,t)\pmod4
\]
gives $\C[2]\subset2r^{\perp_{\bar\beta}}$. Nondegeneracy makes $r^{\perp_{\bar\beta}}$ three-dimensional, so both sides have order eight and are equal. We write $r^\perp$ for this orthogonal complement in the rest of the proof. Since $\C/\C[2]$ has order two and is generated by the image of $v$, one has $\C=\langle v,2r^\perp\rangle$.

For any lift $v$ of $r$ and any $t\in V$, using integral representatives gives
\begin{equation}\label{eq:quadratic-lift}
 q_2(v+2t)\equiv q_2(v)+2\bar\beta(r,t)\pmod4.
\end{equation}
Thus two even lifts $v,v'$ of $r$ differ by an element of $2r^\perp$ and generate the same code with that subgroup. An even lift always exists: a zero-or-one lift has $q_2(v)=0$ or $2$ modulo four; in the second case choose $t$ with $\bar\beta(r,t)=1$ and apply \eqref{eq:quadratic-lift}. Such a $t$ exists because $r\ne0$ and $\bar\beta$ is nondegenerate.

Fix an even lift and form $\C_r=\langle v,2r^\perp\rangle$. The intersection of $\langle v\rangle$ with $2r^\perp$ is $\{0,2v\}$: its even elements are exactly these two, and $r\in r^\perp$ because the polar form is alternating. Consequently
\[
 |\C_r|=\frac{4\cdot8}{2}=16.
\]
For $t\in r^\perp$, one has $q_2(2t)=0$ and $\beta(v,2t)=0$ modulo four. Since $v$ is even, every word satisfies
\[
 q_2(kv+2t)=k^2q_2(v)+q_2(2t)+k\beta(v,2t)=0\pmod4.
\]
Thus $\C_r$ is even; polarization gives isotropy, and its order gives self-duality. Its reduction has dimension one, so it has the required mixed type. Distinct $r$ give distinct reductions and hence distinct codes. There are exactly nine mixed codes.

Finally a subgroup of type $\Z_2^4$ is contained in $2\Z_4^4$ and has the same order. It must equal that subgroup. All pairings and quadratic values of its words vanish modulo four, so it is even self-dual. This gives the unique remaining code.
\end{proof}

\begin{proposition}[Block-monomial orbit decomposition]\label{prop:d2orbits}
The twelve free codes split into three block-monomial orbits of size $4$, the nine mixed codes split into orbits of sizes $4,4,1$, and the two-torsion code is fixed.  Hence the complete orbit-size list is
\[
4,4,4,4,4,1,1.
\]
\end{proposition}

\begin{proof}
The reduction map is equivariant for the block-monomial group. Above each of the four coordinate planes are the two even graph lifts described in Proposition~\ref{prop:d2count}. The $b=0$ graph contains an order-four word with exactly one nonzero coordinate, and the $b=2$ graph does not. For the latter assertion, any order-four graph word has an odd $a$-coordinate; the corresponding entry of $2Ja$ is nonzero in the opposite coordinate pair. Pair permutations, exchanges, and sign changes preserve the number of nonzero coordinates, so the two lifts cannot be equivalent.

The reduced group is transitive on the four coordinate planes. Above every such plane there is exactly one lift with the first property and one without it. Transporting a lift by a block operation therefore proves transitivity on each of these two sets of four codes. Hence these are two orbits of size four.

The two diagonal planes form one reduced orbit. In the stabilizer of the graph plane, interchange the two hyperbolic pairs. On a graph matrix this conjugates $J$ to $-J$ and thus interchanges $b=1$ and $b=3$. Hence the two lifts over one plane are equivalent. An exchange $a_2\leftrightarrow b_2$ then reaches the other diagonal plane, giving a single orbit containing all four diagonal lifts. None can meet a coordinate lift because their reductions belong to different orbits.

For a mixed code, the correspondence $r\mapsto\C_r$ is unique and equivariant. Its three orbits therefore have the same sizes $4,4,1$ as the singular-vector orbits of Lemma~\ref{lem:F2geometry}. The full two-torsion subgroup is fixed by every automorphism of $\Z_4^4$. Their orbit sizes sum to $22$.
\end{proof}

\begin{theorem}\label{thm:d2class}
There are exactly twenty-two even self-dual split codes in $\mathbb Z_4^4$.  Up to block-monomial equivalence there are seven classes, listed in \cref{tab:d2}, with orbit sizes
\[
4,4,4,4,4,1,1.
\]
\end{theorem}

\begin{proof}
Propositions~\ref{prop:d2count} and~\ref{prop:d2orbits} give the counts and orbit sizes. It remains to place the listed generators in those orbits.

For $\mathrm F_0$, exchange both hyperbolic pairs to obtain the graph with $b=0$. For $\mathrm F_1$, exchange $a_2$ and $b_2$; its generators become $0110$ and $1003$, the graph with $b=1$. For $\mathrm F_2$, exchange $a_1$ and $b_1$; its generators become $1002$ and $0120$, the graph with $b=2$. These are the three distinct free orbits.

The order-four generators for $\mathrm M_0$, $\mathrm M_1$, and $\mathrm M_2$ may be taken as $0001$, $0011$, and $1113$. They are even and reduce to $0001$, $0011$, and $1111$. For each row of the table, the doubles of that generator together with the two order-two generators span $2r^\perp$: after division by two they are, respectively,
\[
 \langle0001,0010,1000\rangle,\qquad
 \langle0011,0001,1100\rangle,\qquad
 \langle1111,0011,0101\rangle.
\]
Each displayed subspace has dimension three and is orthogonal to the corresponding $r$. Thus these codes are exactly the three representatives $\C_r$ constructed in Proposition~\ref{prop:d2count}. The generators of $\mathrm T$ span $2\Z_4^4$. All seven entries are therefore even self-dual and represent the seven different orbits.
\end{proof}

\begin{table}[ht]
\centering
\small
\begin{tabular}{c c c c c}
\toprule
Class & generators & group type & orbit size & equivalent graph form\\
\midrule
$\mathrm F_0$ & $\langle0001,0010\rangle$ & $\mathbb Z_4^2$ & $4$ & $b=0$\\
$\mathrm F_1$ & $\langle0011,1300\rangle$ & $\mathbb Z_4^2$ & $4$ & $b=\pm1$\\
$\mathrm F_2$ & $\langle0012,2100\rangle$ & $\mathbb Z_4^2$ & $4$ & $b=2$\\
$\mathrm M_0$ & $\langle0001,0020,2000\rangle$ & $\mathbb Z_4\oplus\mathbb Z_2^2$ & $4$ & no\\
$\mathrm M_1$ & $\langle0002,0011,2200\rangle$ & $\mathbb Z_4\oplus\mathbb Z_2^2$ & $4$ & no\\
$\mathrm M_2$ & $\langle0022,0202,1113\rangle$ & $\mathbb Z_4\oplus\mathbb Z_2^2$ & $1$ & no\\
$\mathrm T$ & $\langle0002,0020,0200,2000\rangle$ & $\mathbb Z_2^4$ & $1$ & no\\
\bottomrule
\end{tabular}
\caption{The seven block-monomial classes for $d=2$.}
\label{tab:d2}
\end{table}

For the three free classes one may use the graph representatives
\begin{equation}\label{eq:graphrep}
\C_b=\{(a,bJa):a\in\mathbb Z_4^2\},
\qquad b=0,1,2,
\end{equation}
with $b=3$ equivalent to $b=1$. A graph vector is the
pair $(a,bJa)$: its second component is the image of its first under
the linear map $bJ$.

\begin{remark}[Finite orthogonal equivalence]\label{rem:finite-orthogonal}
The block-monomial equivalence of Theorem~\ref{thm:d2class} is
smaller than equivalence under the finite orthogonal group
\[
 O(M_2,q_2)=\{g\in\operatorname{Aut}_{\Z_4}(M_2):
 q_2(gv)=q_2(v)\text{ for all }v\in M_2\}.
\]
For $s\in\Z_4$, define
\[
 h_s(a,b)=(a,b+sJa).
\]
Since $a^TJa=0$, the map $h_s$ preserves $q_2$, and hence its
polar form $\beta$. In particular,
$h_s(\C_b)=\C_{b+s}$, so the three free classes lie in one
$O(M_2,q_2)$-orbit. For the mixed classes, take
\[
 D_0=\langle1000,0200,0002\rangle\in\mathrm M_0,
 \qquad D_1=\langle1100,0200,0022\rangle\in\mathrm M_1.
\]
Then
\[
 h_1(D_0)=\langle1003,0220,0002\rangle\in\mathrm M_1,
 \qquad
 h_1(D_1)=\langle1113,0220,0022\rangle\in\mathrm M_2.
\]
Thus the nine mixed codes also form one orbit. Since an automorphism
preserves the abstract group type of a code, the twenty-two codes
form exactly three orbits under $O(M_2,q_2)$, of sizes $12$, $9$, and
$1$, corresponding to $\Z_4^2$, $\Z_4\oplus\Z_2^2$, and $\Z_2^4$.
These are finite-code orbits, not T-duality classes. Indeed, the
lift of $h_s$ to the ambient real space,
\[
 (x,y)\longmapsto(x,y+s(J\otimes I_3)x),
\]
preserves $\eta$ but does not in general preserve $\HH$ in
\eqref{eq:forms}. The seven Narain backgrounds remain inequivalent
by Theorem~\ref{prop:d2Tduality}.
\end{remark}

\begin{remark}[Dependence on the base lattice]\label{rem:base}
Let $R'$ be a positive definite integral lattice with
$R'^*/R'\simeq\Z_4$. The nondegenerate discriminant pairing has the form
$uab/4\pmod\Z$, with $u$ odd, after choosing a generator. For the split
construction on $R'^d\oplus R'^d$, the integrality and evenness conditions
are therefore $u\beta=0$ and $uq_d=0$ modulo four. They are equivalent to
$\beta=0$ and $q_d=0$. Thus the twenty-two codes, the seven block-monomial
classes, and their orbit sizes are unchanged. This does not identify the
resulting CFTs: the rank of $R'$, its Euclidean cosets, and their norms
determine the central charge and spectra. Additional T-dualities must be
examined for each choice of $R'$.

For $R'=2\Z$, with dual $\tfrac12\Z$, the lattice is exactly
\[
 \frac12\{n\in\Z^{2d}:n\bmod4\in\C\},
\]
the ordinary $\Z_4$ Construction A at central charge $d$
\cite[Theorem~4.3]{AlamEtAl}. This rank-one lattice is sometimes denoted
$D_1$; it is not the root lattice $A_1$ in its norm-two normalization.
For odd $m\ge3$, $D_m$ has the same cyclic discriminant group, with
pairing $mab/4\pmod\Z$, and central charge $md$.

The torsion code for one pair occurs in \cite[Eq.~(4.31)]{AlamEtAl}.
That paper also treats two-pair examples. In particular, Case~3 of
\cite[Section~4.3.2, Eq.~(4.62)]{AlamEtAl}, specialized to modulus four, gives
$\{(a,ax,-bx,b):a,b\in\Z_4\}$; $x=3$ gives $\mathrm F_1$ in
Table~\ref{tab:d2}, and $x=2$ gives $\mathrm F_2$ after a pair permutation.
These are particular labeled codes within the classification of
Theorem~\ref{thm:d2class}.
\end{remark}

\section{Metric and \texorpdfstring{$B$}{B}-field representatives}\label{sec:GB}

\subsection{Coordinate convention}
Put $n=3d$. We use the dimensionless Narain convention
\begin{equation}\label{eq:narainform}
 \binom{x}{y}
 =\begin{pmatrix}\gamma^{-T}&\gamma^{-T}B\\0&\gamma\end{pmatrix}
 \binom{k}{\ell},
 \qquad k,\ell\in\Z^n,\quad G=\gamma^T\gamma,\quad B^T=-B.
\end{equation}
Here $B$ is the two-form in the integral winding basis. The displayed basis has split Gram matrix $\left(\begin{smallmatrix}0&I_n\\I_n&0\end{smallmatrix}\right)$. Its positive Gram matrix is
\[
 \begin{pmatrix}G^{-1}&G^{-1}B\\-BG^{-1}&G-BG^{-1}B\end{pmatrix}.
\]
Thus $(G,B)$ determines both forms. The corresponding skew map in
ambient Euclidean coordinates is $\mathcal B=\gamma^{-T}B\gamma^{-1}$.

For an integral antisymmetric matrix $C$, put $B'=B+C$. The same momentum vector has charge coordinates $k'=k-C\ell$ and $\ell'=\ell$ in the primed description, since
\[
 k'+B'\ell'=(k-C\ell)+(B+C)\ell=k+B\ell.
\]
This is an integral unimodular relabelling; its matrix is $\left(\begin{smallmatrix}I&-C\\0&I\end{smallmatrix}\right)$, which preserves the split pairing because $C^T=-C$. For a winding-basis change $A\in GL_n(\Z)$, set
\[
 \gamma'=\gamma A,\qquad B'=A^TBA,\qquad
 k'=A^Tk,\quad\ell'=A^{-1}\ell.
\]
Substitution in \eqref{eq:narainform} again gives the same $(x,y)$, while $G'=A^TGA$. The charge transformation $\diag(A^T,A^{-1})$ is integral unimodular and preserves the split pairing. Equivalence of representatives requires an integral Lorentzian
isometry carrying the source positive form to the target positive form.

\subsection{The free codes}
Let $S$ be the matrix of the simple-root basis of $A_3$ in the $D_3$ realization of Section~\ref{sec:fixed}, and put $S^*=S^{-T}$. Then $S^TS=K$ and $(S^*)^TS^*=K^{-1}$. For the graph representative
\[
 \C_b=\{(a,bJa):a\in\Z_4^2\},\qquad b=0,1,2,
\]
write $L=I_2\otimes S$ and $\gamma=I_2\otimes S^*$. The gluing condition is exactly
\[
 x\in W^2,\qquad y-b(J\otimes I_3)x\in R^2.
\]
Indeed $b(J\otimes I_3)$ preserves $W^2$ and $R^2$ and induces $bJ$ on the discriminant labels. Every allowed point has a unique expression
\[
 x=\gamma\ell,\qquad
 y=Lk+b(J\otimes I_3)\gamma\ell,
 \qquad k,\ell\in\Z^6,
\]
since the columns of $\gamma$ and $L$ are bases of $W^2$ and $R^2$. Interchanging the two null coordinate groups preserves both forms. In the resulting coordinates the lattice therefore has the full integral basis
\[
 \begin{pmatrix}L&b(J\otimes I_3)\gamma\\0&\gamma\end{pmatrix}.
\]
Here $L=\gamma^{-T}$. The upper-right block in \eqref{eq:narainform} must consequently satisfy $\gamma^{-T}B=b(J\otimes I_3)\gamma$, or
\[
 B=\gamma^T b(J\otimes I_3)\gamma=bJ\otimes K^{-1}.
\]
This matrix is antisymmetric, so the basis is null for the split pairing. Its positive form is determined by
\begin{equation}\label{eq:GBgraph}
 G=I_2\otimes K^{-1},\qquad B=bJ\otimes K^{-1}.
\end{equation}
The graph parameter $b=3$ is equivalent to $b=1$.

For one pair, use
\begin{equation}\label{eq:GA}
 G_A=\begin{pmatrix}\tfrac34&-\tfrac12&-\tfrac12\\-\tfrac12&1&0\\-\tfrac12&0&1\end{pmatrix}.
\end{equation}
The columns $(1/2,-1/2,-1/2)^T,e_2,e_3$ lie in $D_3^*$ and have determinant of absolute value $1/2$. This equals the covolume of $D_3^*$, since $D_3$ has covolume two, so they form a basis. Their Gram matrix is $G_A$. The free one-pair lattice is $W\oplus R$, and exchanging the null coordinates gives the basis $\diag(\gamma^{-T},\gamma)$ with these columns of $\gamma$. It has representative $(G_A,0)$. For $\C_2$ both null coordinates range over $W_0\cup W_2=\Z^3$, giving $(I_3,0)$.

\subsection{The nonfree codes}
For the four representatives in Table~\ref{tab:d2}, the Narain data are
\begin{align}
 (G_{\mathrm M_0},B_{\mathrm M_0})&=(I_3\oplus G_A,0),\\
 (G_{\mathrm M_1},B_{\mathrm M_1})&=(G_M,0),\\
 (G_{\mathrm M_2},B_{\mathrm M_2})&=(G_M,B_M),\\
 (G_{\mathrm T},B_{\mathrm T})&=(I_6,0),
\end{align}
where
\begin{equation}\label{eq:GM}
 G_M=\begin{pmatrix}
 \tfrac32&-\tfrac12&-\tfrac12&-\tfrac12&-\tfrac12&-\tfrac12\\
 -\tfrac12&1&0&0&0&0\\-\tfrac12&0&1&0&0&0\\
 -\tfrac12&0&0&1&0&0\\-\tfrac12&0&0&0&1&0\\-\tfrac12&0&0&0&0&1
 \end{pmatrix},
\end{equation}
and, in the winding-basis convention of \eqref{eq:narainform},
\begin{equation}\label{eq:BM}
 B_M=\frac12\begin{pmatrix}
 0&-1&-1&-1&-1&-1\\1&0&0&0&0&0\\1&0&0&0&0&0\\
 1&0&0&0&0&0\\1&0&0&0&0&0\\1&0&0&0&0&0
 \end{pmatrix}.
\end{equation}
To prove these formulas, work with the numerator lattice
\[
 \mathcal L_{\C}=\{z\in\Z^{12}:z\bmod4\in\widehat\C\},
 \qquad \Lam_{\cA}(\widehat\C)=\tfrac12\mathcal L_{\C}.
\]
For a gluing code $\C\subset\Z_4^2\oplus\Z_4^2$, put
\[
 A_{\C}=\{a\in\Z_4^2:(a,0)\in\C\}.
\]
The pure first-null part of the numerator lattice is the inverse image of
$A_{\C}$ under the two-block quotient map. For the four nonfree classes,
\[
 A_{\mathrm M_0}=\{00,20\},\qquad
 A_{\mathrm M_1}=A_{\mathrm M_2}=\{00,22\},\qquad
 A_{\mathrm T}=\{00,02,20,22\}.
\]
Set
\[
 H_A=\begin{pmatrix}4&2&2\\0&2&0\\0&0&2\end{pmatrix},
 \qquad
 H_M=\begin{pmatrix}
 4&2&2&2&2&2\\0&2&0&0&0&0\\0&0&2&0&0&0\\
 0&0&0&2&0&0\\0&0&0&0&2&0\\0&0&0&0&0&2
 \end{pmatrix},
\]
and take
\[
 H_{\mathrm M_0}=2I_3\oplus H_A,\qquad
 H_{\mathrm M_1}=H_{\mathrm M_2}=H_M,\qquad H_{\mathrm T}=2I_6.
\]
The column labels follow from \eqref{eq:blocklabel}. In $H_A$, the columns are $4e_1$, $2(e_1+e_2)$, and $2(e_1+e_3)$, all of label zero. Thus the columns of $H_{\mathrm M_0}$ have labels $20$ in its first three positions and $00$ in the last three. In $H_M$, the first three columns have labels $00$ and the last three have labels $22$. The columns of $H_{\mathrm T}$ have labels $20$ or $02$. These are all in the respective groups $A_{\C}$ above.

The two-block quotient map has kernel $E^2$ of order $|E|^2$. Hence the residue subgroup of the pure numerator lattice has order $|E|^2|A_{\C}|$, and its inverse image in $\Z^6$ has index
\[
 \frac{4^6}{|E|^2|A_{\C}|}.
\]
This index is $128$ for $\mathrm M_0,\mathrm M_1,\mathrm M_2$ and $64$ for $\mathrm T$. The triangular matrices $H_{\C}$ have exactly these absolute determinants. Their column lattices are contained in the pure numerator lattices and have the same indices in $\Z^6$. The matching indices prove that the displayed columns are bases of those lattices.

Put
\[
 Y_{\C}=4H_{\C}^{-T}.
\]
The required inverse transposes are
\[
 4H_A^{-T}=\begin{pmatrix}1&0&0\\-1&2&0\\-1&0&2\end{pmatrix},
 \qquad
 4H_M^{-T}=\begin{pmatrix}1&0_{1\times5}\\-\mathbf1_5&2I_5\end{pmatrix},
 \qquad \mathbf1_5=(1,1,1,1,1)^T.
\]
Together with $Y_{\mathrm M_0}=2I_3\oplus4H_A^{-T}$ and $Y_{\mathrm T}=2I_6$, these show integrality and give $H_{\C}^TY_{\C}=4I_6$. For $\mathrm M_0,\mathrm M_1,\mathrm T$ take $X_{\C}=0$. For $\mathrm M_2$, multiplication gives
\[
 X_{\mathrm M_2}=H_MB_M=
 \begin{pmatrix}5&-2\mathbf1_5^T\\\mathbf1_5&0_{5\times5}\end{pmatrix},
\]
which is integral despite the half-integral entries of $B_M$. In this case
\[
 X_{\C}^TY_{\C}+Y_{\C}^TX_{\C}
 =4B_M^T+4B_M=0.
\]
For the other three classes, the identity holds because $X_{\C}=0$.
Thus the columns of
\begin{equation}\label{eq:NCbasis}
 \mathcal N_{\C}=\begin{pmatrix}H_{\C}&X_{\C}\\0&Y_{\C}\end{pmatrix}
\end{equation}
have split Gram matrix $4\left(\begin{smallmatrix}0&I_6\\I_6&0\end{smallmatrix}\right)$.
The first six columns belong to the numerator lattice by the definition of $H_{\C}$. For the last six columns, the discriminant labels are
\[
\begin{array}{c|l}
\C&\text{labels of }(X_{\C}e_j,Y_{\C}e_j),\ 1\le j\le6\\\hline
\mathrm M_0&0020,0020,0020,0001,0002,0002\\
\mathrm M_1&0013,0020,0020,0002,0002,0002\\
\mathrm M_2&1113,2020,2020,2002,2002,2002\\
\mathrm T&0020,0020,0020,0002,0002,0002.
\end{array}
\]
The labels are obtained by applying \eqref{eq:blocklabel} to the four blocks of each column. For $\mathrm M_0$ and $\mathrm T$ the resulting words are among their generators or their doubles. For $\mathrm M_1$, the only additional checks are
\[
 0013=0011+0002,\qquad 0020=2(0011)+0002.
\]
For $\mathrm M_2$ they are
\[
 2020=2(1113)+0202,\qquad
 2002=2(1113)+0022+0202.
\]
All equalities here are modulo four. Thus every complementary column also belongs to $\mathcal L_{\C}$. Finally,
\[
 |\det\mathcal N_{\C}|=|\det H_{\C}|\,|\det(4H_{\C}^{-T})|=4^6,
\]
which is exactly the index of the numerator lattice because $|\widehat\C|=4^6$. Equality of indices gives equality of these lattices. Dividing \eqref{eq:NCbasis} by two gives a full integral null basis.

Set $\gamma=2H_{\C}^{-T}$. Then $G=\gamma^T\gamma$ gives the metrics above, while $X_{\C}=H_{\C}B$ gives $B=0$ for $\mathrm M_0,\mathrm M_1,\mathrm T$ and \eqref{eq:BM} for $\mathrm M_2$. 

Any other null complement dual to $H_{\C}/2$ has the same lower block $Y_{\C}/2$, because the dual-pairing equation is $H_{\C}^TY=4I_6$. The difference of its upper block from $X_{\C}/2$ lies in the pure first-null lattice. Hence $X'_{\C}-X_{\C}=H_{\C}C$ with $C$ integral. The null condition forces $C^T=-C$, so the two winding-basis forms differ by the integral antisymmetric matrix $C$.

For $\mathrm M_2$, $B$ cannot be integral with this first-null basis. The first column of
$Y_{\mathrm M_2}$ has second-half label $13$. In $\C_{\mathrm M_2}$ a
word with second half $13$ has first-half label $11$ or $33$, whereas
the pure first-null subgroup has labels only $00$ and $22$. If $B$ were
integral, then $X=H_MB$ would lie in the pure first-null lattice and its
first-half label would belong to $\{00,22\}$, a contradiction. The
class of $B_M$ modulo integral antisymmetric matrices therefore has
order exactly two, since $2B_M$ is integral. It remains nonzero after an
integral change of winding basis: if $A^TB_MA$ were integral for
$A\in GL_6(\Z)$, then $B_M=A^{-T}(A^TB_MA)A^{-1}$ would be integral.
The order of this matrix class refers to the chosen null polarization.

\subsection{The obstruction to a zero \texorpdfstring{$B$}{B}-field}

\begin{proposition}\label{prop:essentialB}
Among the seven two-pair classes, $\mathrm M_2$ is the only one with no
T-duality representative having $B=0$. In fact no representative of
$\mathrm M_2$ has integral $B$.
\end{proposition}
\begin{proof}
First consider the other six representatives in Table~\ref{tab:d2}.
They have the product form $P\times P^\perp$, where the orthogonal
complement is taken for the ordinary dot product on $\Z_4^2$:
\[
\begin{array}{c|c|c}
\text{class}&P&P^\perp\\\hline
\mathrm F_0&0&\Z_4^2\\
\mathrm F_1&\langle13\rangle&\langle11\rangle\\
\mathrm F_2&\langle21\rangle&\langle12\rangle\\
\mathrm M_0&\langle20\rangle&\langle01,20\rangle\\
\mathrm M_1&\langle22\rangle&\langle11,02\rangle\\
\mathrm T&2\Z_4^2&2\Z_4^2.
\end{array}
\]
Let $L(P)$ be the inverse image of $P$ in $W^2$. The discriminant
pairing is $3a\cdot b/4\pmod\Z$, so
$L(P)^*=L(P^\perp)$. Hence the momentum lattice is
$L(P)\oplus L(P)^*$ in the two null coordinate groups. A basis $\gamma$
of $L(P)^*$ and its dual basis $\gamma^{-T}$ give
\eqref{eq:narainform} with $B=0$. In particular, the fractional
$B$-fields of the graph representatives of $\mathrm F_1$ and
$\mathrm F_2$ are removable by T-duality.

For $\mathrm M_2$, use the twelve Euclidean coordinates consisting of
the six $x$-coordinates followed by the six $y$-coordinates. Set
\begin{equation}\label{eq:D12extension}
 D_{12}=\{z\in\Z^{12}:\textstyle\sum_jz_j\equiv0\pmod2\},
 \qquad s=\tfrac12(1,1,1,1,1,1,1,1,1,1,1,-1).
\end{equation}
Then
\begin{equation}\label{eq:M2positive}
 \Lam_{A_3}(\mathrm M_2)=D_{12}\ \cup\ (s+D_{12}).
\end{equation}
Indeed, the words with even entries in $\mathrm M_2$ are twice the
binary even-weight words of length four. An integral triple has label
zero or two according to the parity of its coordinate sum, so their
inverse image is $D_{12}$. The vector $s$ has labels $1113$ by
\eqref{eq:blocklabel}; the remaining codewords are its translates of
the even subgroup. This proves \eqref{eq:M2positive}.

The minimum squared norm for $\HH$ is two: the integral part has
minimum two, whereas every vector in the
half-integral coset has squared norm at least $12/4=3$. Its norm-two
vectors are exactly the roots $\pm e_i\pm e_j$ of $D_{12}$, which span
$\R^{12}$ and form an irreducible root system. These are roots of the
whole positive lattice, not the chiral root systems computed below.

T-duality preserves the momentum lattice with both forms, hence its
positive isometry class. If a frame had $B=0$, its positive Gram matrix would be
$\diag(G^{-1},G)$ by \eqref{eq:narainform}. The positive lattice would
then be an orthogonal direct sum of two rank-six lattices. In such a
decomposition a norm-two vector cannot have nonzero components in
both factors, since each nonzero component is a lattice vector of
squared norm at least two. Every $D_{12}$ root must therefore belong
to one factor. Nonorthogonal roots must belong to the same factor,
and irreducibility forces all the roots into one factor. This is
impossible because they span rank twelve. Finally, an integral
antisymmetric $B$ can be removed by an integral $B$-shift, so an
integral representative is also impossible.
\end{proof}

Absence of a zero-$B$ frame is T-duality invariant. The order-two class
of $B_M$ in \eqref{eq:BM} is defined in that polarization, without a
canonical identification across all polarizations.

\section{Theta series, spectra and T-duality}\label{sec:spectra}

\subsection{Exact theta formulas}
For $\tau$ in the upper half-plane put $q=e^{2\pi i\tau}$ and $\bar q=e^{-2\pi i\bar\tau}$. With the cosets from \eqref{eq:cosets}, define
\begin{equation}\label{eq:blocktheta}
 \Theta_{ab}(\tau,\bar\tau)
 =\sum_{x\in W_a}\sum_{y\in W_b}
 q^{|x+y|^2/4}\bar q^{|x-y|^2/4}.
\end{equation}
The coset decomposition gives
\begin{equation}\label{eq:thetaC}
 \Theta_{\C}(\tau,\bar\tau)
 =\sum_{(a,b)\in\C}\prod_{i=1}^d\Theta_{a_i b_i}(\tau,\bar\tau).
\end{equation}
Equivalently, \eqref{eq:thetaC} evaluates the paired-symbol weight enumerator $\sum_{(a,b)\in\C}\prod_iX_{a_i b_i}$ at $X_{ab}=\Theta_{ab}$. The torus partition function is
\begin{equation}\label{eq:partition}
 Z_{\C}(\tau,\bar\tau)=\frac{\Theta_{\C}(\tau,\bar\tau)}{|\eta_D(\tau)|^{6d}},
 \qquad \eta_D(\tau)=q^{1/24}\prod_{n\ge1}(1-q^n).
\end{equation}
The denominator accounts for the left and right oscillators of $3d$ compact bosons. In numerator coordinates $x=n/2$, $y=m/2$,
\begin{equation}\label{eq:hexponents}
 h=|n+m|^2/16,\qquad\bar h=|n-m|^2/16,
 \qquad \Delta=h+\bar h=(|n|^2+|m|^2)/8.
\end{equation}

Poisson summation on the sixteen cosets gives the modular transformations of the block functions.

\begin{proposition}\label{prop:modular}
The block theta functions satisfy
\begin{align}
 \Theta_{ab}(\tau+1,\bar\tau+1)
 &=e^{2\pi i\,3ab/4}\Theta_{ab}(\tau,\bar\tau),
 \label{eq:thetaT}\\
 \Theta_{ab}(-1/\tau,-1/\bar\tau)
 &=\frac{|\tau|^3}{4}\sum_{c,e\in\Z_4}
 e^{-2\pi i\,3(ae+bc)/4}\Theta_{ce}(\tau,\bar\tau).
 \label{eq:thetaS}
\end{align}
For every even self-dual gluing code $\C$, the function
$Z_{\C}$ in \eqref{eq:partition} is invariant under
$\tau\mapsto\tau+1$ and $\tau\mapsto-1/\tau$.
\end{proposition}
\begin{proof}
For $x\in W_a$ and $y\in W_b$, the discriminant pairing gives $x\cdot y\equiv3ab/4\pmod\Z$. Since $h-\bar h=x\cdot y$, shifting both arguments by one multiplies every summand by $e^{2\pi i\,3ab/4}$. This proves \eqref{eq:thetaT}.

For the second transformation, write $\tau=\tau_1+i\tau_2$ and set
\[
 f_\tau(x,y)=\exp\bigl(-\pi\tau_2(|x|^2+|y|^2)
                       +2\pi i\tau_1x\cdot y\bigr).
\]
Use the Fourier-transform convention
\[
 \widehat f_\tau(u,v)=\int_{\R^3\oplus\R^3}
 f_\tau(x,y)e^{-2\pi i(x\cdot v+y\cdot u)}\,dx\,dy.
\]
For each coordinate pair the quadratic matrix and its inverse are
\[
 A_\tau=\begin{pmatrix}\tau_2&-i\tau_1\\-i\tau_1&\tau_2\end{pmatrix},
 \qquad
 A_\tau^{-1}=\frac1{|\tau|^2}
     \begin{pmatrix}\tau_2&i\tau_1\\i\tau_1&\tau_2\end{pmatrix}.
\]
Its real part is positive definite and $\det A_\tau=|\tau|^2$. The Gaussian integral therefore gives a factor $|\tau|^{-1}$ per pair. The inverse matrix, with the two frequency coordinates exchanged, is the quadratic matrix for $-1/\tau$. Taking all three pairs yields
\[
 \widehat f_\tau(u,v)=|\tau|^{-3}f_{-1/\tau}(u,v).
\]

Let $L=D_3\oplus D_3$. Its covolume is four and its dual for the split pairing is $L^\#=D_3^*\oplus D_3^*$. Representatives for $L^\#/L$ are $\lambda_{ab}=(ag/2,bg/2)$, with
\[
 \eta(\lambda_{ab},\lambda_{ce})
       \equiv\frac{3(ae+bc)}4\pmod\Z.
\]
Poisson summation on $L+\lambda_{ab}$ gives
\[
 \Theta_{ab}(\tau,\bar\tau)
 =\frac1{4|\tau|^3}\sum_{c,e\in\Z_4}
 e^{2\pi i\,3(ae+bc)/4}
 \Theta_{ce}(-1/\tau,-1/\bar\tau).
\]
To invert this finite transform, note that
\[
 \sum_{c,e\in\Z_4}
 e^{2\pi i\,3((a-a')e+(b-b')c)/4}
       =16\delta_{aa'}\delta_{bb'}.
\]
Each factor is a four-term geometric sum, and three is invertible modulo four. Multiplication by the transform with the opposite phase now gives exactly \eqref{eq:thetaS}, including its factor $|\tau|^3/4$.

In the full code sum \eqref{eq:thetaC}, the $T$-phase of a summand is $e^{2\pi i\,3q_d(a,b)/4}=1$. After the $S$-transformation, the coefficient of the product with label $z\in\Z_4^{2d}$ is
\[
 \frac{|\tau|^{3d}}{4^d}
 \sum_{c\in\C}e^{-2\pi i\,3\beta(c,z)/4}.
\]
This character sum equals $|\C|$ when $z\in\C^\perp$ and zero otherwise. In the latter case, translation by a codeword on which the character is nontrivial multiplies the sum by a number different from one, forcing it to vanish. Since $\C^\perp=\C$ and $|\C|=4^d$, the theta sum transforms by $|\tau|^{3d}$. The identities
\[
 |\eta_D(-1/\tau)|=|\tau|^{1/2}|\eta_D(\tau)|,
 \qquad |\eta_D(\tau+1)|=|\eta_D(\tau)|
\]
show that the denominator of \eqref{eq:partition} cancels this factor and is unchanged by $T$. Thus $Z_{\C}$ is invariant under both generators.
\end{proof}

\subsection{One-variable theta identities}
For the one-variable momentum series, write
\[
 u(t)=\sum_{x\in W_0}t^{x^2/2},\quad
 v(t)=\sum_{x\in W_2}t^{x^2/2},\quad
 w(t)=\sum_{x\in W_1}t^{x^2/2}=\sum_{x\in W_3}t^{x^2/2}.
\]
The equality follows from $W_3=-W_1$. The first terms are
\begin{align}\label{eq:oneblock}
 u&=1+12t+6t^2+O(t^3),\notag\\
 v&=6t^{1/2}+8t^{3/2}+O(t^{5/2}),\\
 w&=4t^{3/8}+12t^{11/8}+O(t^{19/8}).\notag
\end{align}
For $u$, the first terms count the zero vector, the twelve vectors $\pm e_i\pm e_j$, and the six vectors $\pm2e_i$. Every vector of $W_0=D_3$ has even squared norm, so the next possible exponent is three. The coset $W_2$ consists of the integer triples with odd coordinate sum. Their squared norms are odd; at norms one and three the triples are the six unit vectors and the eight choices $(\pm1,\pm1,\pm1)$. The next possible squared norm is five.

In the half-integral cosets, the shortest triples have all coordinate absolute values $1/2$. Of the eight sign choices, four have each label, by \eqref{eq:blocklabel}; their squared norm is $3/4$. At the next squared norm $11/4$, exactly one coordinate has absolute value $3/2$. There are $3\cdot8$ such triples, divided equally between the two labels by negation. The next possible norm is $19/4$, obtained by enlarging two coordinates to absolute value $3/2$. This proves \eqref{eq:oneblock}, including the remainder bounds.

Put $\vartheta_0=u$, $\vartheta_2=v$ and $\vartheta_1=\vartheta_3=w$. Since
$\Delta=(x^2+y^2)/2$, the coset decomposition gives the exact one-variable identity
\begin{equation}\label{eq:univcode}
 \Theta_{\C}(t)=\sum_{c=(c_1,\ldots,c_{2d})\in\C}
 \prod_{j=1}^{2d}\vartheta_{c_j}(t).
\end{equation}
For one pair, the codes $\C_x$ and $\C_2$ give
\begin{equation}\label{eq:onedtheta}
 \Theta_{\mathrm F}=u(u+v+2w),\qquad
 \Theta_{\mathrm T}=(u+v)^2.
\end{equation}
For two pairs, \eqref{eq:univcode} gives the polynomials in
Table~\ref{tab:thetapolys}.

The words of $\mathrm F_0$ are $(0,0,s,t)$, with $s,t\in\Z_4$. The two zero labels contribute $u^2$, while each independent label contributes $u+v+2w$. For $\mathrm F_1$ the words are $(t,3t,s,s)$. Each of its two label pairs contributes $u^2$ at parameter zero, $v^2$ at parameter two, and $2w^2$ from the two odd parameters. Thus its theta polynomial is $(u^2+v^2+2w^2)^2$. The words of $\mathrm F_2$ are $(2t,t,s,2s)$. For either pair, the four parameter values contribute $u^2$, $vw$, $uv$, and $vw$, respectively, giving $(u^2+uv+2vw)^2$.

For $\mathrm M_0$ the words are $(2a,0,2b,c)$, with $a,b\in\Z_2$ and $c\in\Z_4$. Summing over these three parameters gives $u(u+v)^2(u+v+2w)$. For $\mathrm M_1$ they are $(2a,2a,c,c+2b)$. The first pair contributes $u^2+v^2$. When $c$ is even, the second pair runs through all four zero-or-two pairs and contributes $(u+v)^2$; when $c$ is odd, the four pairs contribute $4w^2$.

Write a word of $\mathrm M_2$ as
\[
 (k,k+2b,k+2a,3k+2a+2b),\qquad
 k\in\Z_4,\quad a,b\in\Z_2.
\]
For $k=2c$ the word divided by two is
$(c,c+b,c+a,c+a+b)\in\F_2^4$. These are the eight binary words of even weight, so the even words contribute $u^4+6u^2v^2+v^4$. The other eight words have all four labels odd and contribute $8w^4$. Finally the four coordinates of $\mathrm T$ range independently over $\{0,2\}$, giving $(u+v)^4$.
\begin{table}[ht]
\centering
\begin{tabular}{c l}
\toprule
Class & exact one-variable theta series\\\midrule
$\mathrm F_0$ & $u^2(u+v+2w)^2$\\
$\mathrm F_1$ & $(u^2+v^2+2w^2)^2$\\
$\mathrm F_2$ & $(u^2+uv+2vw)^2$\\
$\mathrm M_0$ & $u(u+v)^2(u+v+2w)$\\
$\mathrm M_1$ & $(u^2+v^2)((u+v)^2+4w^2)$\\
$\mathrm M_2$ & $u^4+6u^2v^2+v^4+8w^4$\\
$\mathrm T$ & $(u+v)^4$\\\bottomrule
\end{tabular}
\caption{Theta identities for the representatives in Table~\ref{tab:d2}.}
\label{tab:thetapolys}
\end{table}

\subsection{The one-pair spectra}
The two series in \eqref{eq:onedtheta} begin
\begin{align}
\Theta_{\mathrm F}
&=1+8t^{3/8}+6t^{1/2}+24t+120t^{11/8}+80t^{3/2}+156t^2+O(t^{17/8}),\label{eq:d1free}\\
\Theta_{\mathrm T}
&=1+12t^{1/2}+60t+160t^{3/2}+252t^2+O(t^{5/2}).\label{eq:d1tors}
\end{align}
Substitution of \eqref{eq:oneblock} into \eqref{eq:onedtheta} gives
these series. Their distinct momentum-primary gaps, $3/8$ and $1/2$,
exclude T-duality between the two backgrounds. The torsion class is the square point $(G,B)=(I_3,0)$ of Section~\ref{sec:GB}, while the free class has representative $(G_A,0)$.

\subsection{The two-pair spectra}
Substitution of \eqref{eq:oneblock} into Table~\ref{tab:thetapolys} gives
Table~\ref{tab:spectra}. Its coefficients count momentum states
through $\Delta=2$, without oscillator excitations. The oscillator
contribution is included in \eqref{eq:partition}.

\begin{table}[ht]
\centering
\scriptsize
\renewcommand{\arraystretch}{1.15}
\begin{tabular}{c c p{9.7cm}}
\toprule
Class & gap & $\Theta_{\C}(t)$ through $\Delta=2$\\
\midrule
$\mathrm F_0$ & $3/8$ & $1+16t^{3/8}+12t^{1/2}+64t^{3/4}+96t^{7/8}+84t+624t^{11/8}+448t^{3/2}+1920t^{7/4}+2720t^{15/8}+1848t^2$\\
$\mathrm F_1$ & $3/4$ & $1+64t^{3/4}+120t+1024t^{3/2}+4224t^{7/4}+4104t^2$\\
$\mathrm F_2$ & $1/2$ & $1+12t^{1/2}+96t^{7/8}+84t+576t^{11/8}+448t^{3/2}+2304t^{7/4}+2720t^{15/8}+1848t^2$\\
$\mathrm M_0$ & $3/8$ & $1+8t^{3/8}+18t^{1/2}+96t^{7/8}+156t+600t^{11/8}+888t^{3/2}+2720t^{15/8}+3768t^2$\\
$\mathrm M_1$ & $1/2$ & $1+12t^{1/2}+64t^{3/4}+120t+880t^{3/2}+4224t^{7/4}+4104t^2$\\
$\mathrm M_2$ & $1$ & $1+264t+2048t^{3/2}+7944t^2$\\
$\mathrm T$ & $1/2$ & $1+24t^{1/2}+264t+1760t^{3/2}+7944t^2$\\
\bottomrule
\end{tabular}
\caption{Low-lying momentum-primary spectra for the seven $d=2$ classes.  The variable $t$ records $\Delta=h+\bar h$.}
\label{tab:spectra}
\end{table}

The lowest nonconstant terms of $\Theta_{\C}(q,\bar q)$ are
\begin{center}
\begin{tabular}{c c}
\toprule
Class & first bivariate shell\\
\midrule
$\mathrm F_0$ & $16q^{3/16}\bar q^{3/16}$\\
$\mathrm F_1$ & $64q^{3/8}\bar q^{3/8}$\\
$\mathrm F_2$ & $12q^{1/4}\bar q^{1/4}$\\
$\mathrm M_0$ & $8q^{3/16}\bar q^{3/16}$\\
$\mathrm M_1$ & $12q^{1/4}\bar q^{1/4}$\\
$\mathrm M_2$ & $12q+240q^{1/2}\bar q^{1/2}+12\bar q$\\
$\mathrm T$ & $24q^{1/4}\bar q^{1/4}$\\
\bottomrule
\end{tabular}
\end{center}
Evenness of the Lorentzian lattice determines the spins in these first shells:
\[
 h-\bar h=\eta(z,z)/2\in\Z.
\]
Since $h,\bar h\ge0$, a momentum state of dimension $\Delta<1$ must have spin zero, and hence $h=\bar h=\Delta/2$. This determines the first shells of every class except $\mathrm M_2$ directly from Table~\ref{tab:spectra}. For $\mathrm M_2$ the gap is one and the possible spins are $-1,0,1$. The current calculation in Proposition~\ref{prop:currents} gives twelve states of each nonzero spin; subtracting these from the total coefficient $264$ leaves $240$ states at $(h,\bar h)=(1/2,1/2)$.

\subsection{Currents and equivalence classes}
\begin{proposition}\label{prop:currents}

For every $d$, put
\[
 t_i=(2e_i,2e_i)\in M_d,\qquad
 m(\C)=\#\{i\in\{1,\ldots,d\}:t_i\in\C\},
\]
where $e_i$ is the $i$th coordinate vector of $\Z_4^d$.
The left and right chiral root systems are both $A_1^{3m(\C)}$.
The finite-dimensional weight-one Lie algebra on each side is
$\mathfrak{su}(2)^{\oplus3m(\C)}\oplus
\mathfrak u(1)^{\oplus(3d-3m(\C))}$; its currents generate the affine algebra
\begin{equation}\label{eq:allcurrents}
 \widehat{\mathfrak{su}}(2)_1^{\,\oplus3m(\C)}
 \ \oplus\ \widehat{\mathfrak u}(1)^{\,\oplus(3d-3m(\C))}.
\end{equation}
In particular no member of this fixed-metric split family has an
$\mathfrak{su}(4)$ current algebra.

For $d=1$, the free class has no chiral norm-two momentum vectors, whereas the torsion class has root system $A_1^3$ on each side. For $d=2$, the chiral root systems are
\[
\begin{array}{c|c|c}
\text{class}&\text{number of roots on each side}&\text{root system}\\\hline
\mathrm F_0,\mathrm F_1,\mathrm F_2&0&\varnothing\\
\mathrm M_0&6&A_1^3\\
\mathrm M_1&0&\varnothing\\
\mathrm M_2,\mathrm T&12&A_1^6
\end{array}
\]
\end{proposition}
\begin{proof}
For a state of weights $(1,0)$, the formulas
$h=|p_L|^2/2+N$ and $\bar h=|p_R|^2/2+\bar N$, with
$N,\bar N\in\Z_{\ge0}$, force $p_R=0$ and $\bar N=0$.
Either $N=1$ and the momentum vanishes, giving one of the $3d$
oscillator currents, or $N=0$ and $|p_L|^2=2$.

A left-chiral momentum vector satisfies $p_R=0$, so $n=m$ in \eqref{eq:hexponents}. Its squared left norm is $|n|^2/2$, so a root has $|n|^2=4$. Each length-three block of $n$ has constant parity because its residue lies in $E^\perp$. An integer vector of squared norm four is either $\pm2e_j$ or has exactly four entries of absolute value one. In the second case all odd entries would have to occur in whole blocks of length three. Their number would be a multiple of three, not four, which is impossible. Thus every root comes from $n=\pm2e_j$, for every $d$.

The occupied block has label two by \eqref{eq:blocklabel}, and all other blocks have label zero. Thus an admitted word $t_i$ gives six roots in three orthogonal directions, yielding $3m(\C)$ factors of $A_1$. Their Cartan currents occupy $3m(\C)$ oscillator directions; the others remain abelian. No $\mathfrak{su}(4)$ embeds in a sum of $\mathfrak{su}(2)$ and abelian factors. For one pair the required word is $22$, which belongs to $\C_2$ and to neither $\C_x$ nor $\C_y$.

For two pairs the required words are $2020$ and $0202$. The free graph codes admit neither: for $a=(2,0)$ or $(0,2)$, the nonzero entry of $bJa$, if any, is in the other coordinate. For $\mathrm M_0$, its word description $(2a,0,2b,c)$ admits $2020$ but excludes $0202$. For $\mathrm M_1$, equality of its first two coordinates excludes both. For $\mathrm M_2$, both are even words with two entries equal to two and hence belong to its even-word subgroup described above. They also belong to $\mathrm T$. This proves the table of left-moving root counts. On the right one has $n=-m$, but negation does not change label two; the same tests and counts apply.

For these roots $p_L=n/\sqrt2=\pm\sqrt2e_j$, or the corresponding right-moving vector. With the usual free-boson normalization $\phi_j(z)\phi_k(w)\sim-\delta_{jk}\log(z-w)$, the fields
\[
 H_j=\frac{i}{\sqrt2}\partial\phi_j,
 \qquad E_j^\pm={:}e^{\pm i\sqrt2\phi_j}{:}
\]
satisfy $H_j(z)H_j(w)\sim1/(2(z-w)^2)$, $H_j(z)E_j^\pm(w)\sim\pm E_j^\pm(w)/(z-w)$, and
\[
 E_j^+(z)E_j^-(w)\sim\frac1{(z-w)^2}+\frac{2H_j(w)}{z-w}.
\]
These are the level-one $\widehat{\mathfrak{su}}(2)$ relations, after the standard normalization of lattice-vertex cocycles. Products from different directions have no singular cross terms. The remaining oscillator directions give the abelian summands in \eqref{eq:allcurrents}; a zero entry in the table excludes only additional momentum currents.
\end{proof}

The gluing lattice $A_3$ is not the chiral momentum lattice
$\{\sqrt2x:(x,x)\in\Lam\}$. Indeed, for a root $\alpha\in R$,
the vector $(\alpha,0)$ has
\[
 p_L=p_R=\alpha/\sqrt2,\qquad (h,\bar h)=(1/2,1/2).
\]
This is a nonchiral scalar of dimension one. Proposition~\ref{prop:currents}
excludes the undeformed level-one $\mathrm{SU}(4)^d$ WZW backgrounds
at the same central charge, whose $A_3$ roots are chiral \cite{Angelinos}.
The statement concerns the fixed positive form used here and does not
exclude reaching other enhanced points by continuous Narain deformation.

\begin{theorem}[T-duality classes in the two-pair family]\label{prop:d2Tduality}
Among the Narain backgrounds arising from the split $d=2$ $A_3$ gluing construction, there are exactly seven equivalence classes under $O(6,6;\Z)$. Representatives are the seven codes in Table~\ref{tab:d2}.
\end{theorem}
\begin{proof}
Each generator of block-monomial equivalence is induced by a real linear isometry preserving both forms in \eqref{eq:forms}. A pair permutation permutes the corresponding $x$- and $y$-blocks together; a sign change negates both blocks; and an exchange interchanges $x_i$ with $y_i$. These maps preserve $R$ and $W$ in each block and carry the gluing condition to that of the transformed code. They therefore restrict to bijective lattice maps of the full Narain backgrounds. Relative to integral null bases in the source and target, such a map and its inverse have integral matrices, and preservation of $\eta$ gives $S^T\eta_0S=\eta_0$. Hence the coordinate matrix belongs to $O(6,6;\Z)$, with the positive form carried to the target positive form. The seven code orbits thus give at most seven Narain classes, and Theorem~\ref{thm:d2class} proves that every background in this family is represented.

Conversely, preservation of $\eta$ and $\HH$ preserves
\[
 h=\frac{\HH(z,z)+\eta(z,z)}4,
 \qquad \bar h=\frac{\HH(z,z)-\eta(z,z)}4
\]
for each momentum vector. In particular it preserves the one-variable theta series. Table~\ref{tab:spectra} distinguishes all seven representatives. The only repeated gaps are $3/8$ and $1/2$: at gap $3/8$, the leading coefficients of $\mathrm F_0$ and $\mathrm M_0$ are $16$ and $8$; at gap $1/2$, $\mathrm T$ has leading coefficient $24$, while $\mathrm F_2$ and $\mathrm M_1$ have coefficient $12$ but differ at exponent $3/4$, with coefficients $0$ and $64$. The gaps of $\mathrm F_1$ and $\mathrm M_2$ are different from these and from one another. Thus no two representatives are T-dual.
\end{proof}

Equality of spectra need not imply T-duality; isospectral inequivalent
code theories occur in \cite{DymarskyShapere}.

\section{Finite shift orbifolds}\label{sec:orbifolds}

\subsection{A pair of gluing codes}
A shift symmetry acts on momentum vertex operators by a character and
fixes the oscillator currents. The lattice formulation, including
twisted-sector phases, is given in
\cite[Section~3]{IsraelMelnikovProto}.

\begin{proposition}\label{prop:orbifold}
Let $\C,\C'\subset M_d$ be even self-dual codes and put
$I=\C\cap\C'$. There is a finite shift orbifold from the background
of $\C$ to that of $\C'$, with gauged group
\begin{equation}\label{eq:gaugegroup}
 G=\C'/I\simeq\operatorname{Hom}(\C/I,U(1)).
\end{equation}
For $v\in\C'$, the corresponding symmetry on a state with code label
$c\in\C$ is
\begin{equation}\label{eq:finitecharacter}
 \chi_v(c)=\exp\!\left(\frac{2\pi i}{4}\,3\beta(c,v)\right).
\end{equation}
The target code specifies the projected twisted sectors as well as
this action on the original theory.
\end{proposition}
\begin{proof}
The pairing in \eqref{eq:finitecharacter} is trivial on $I$ in either
argument. Its kernel on $\C$ is $\C\cap(\C')^\perp=I$, and its kernel
on $\C'$ is $\C'\cap\C^\perp=I$. Since $|\C|=|\C'|$, it induces a
perfect pairing between their quotients. Thus the abstract groups
$\C/I$ and $\C'/I$ are isomorphic, but the acting group is naturally
the character dual of $\C/I$.

Write $\Lam=\Lam_{A_3}(\C)$, $\Lam'=\Lam_{A_3}(\C')$, and
$L=\Lam\cap\Lam'$. A representative $u\in\Lam'$ of a class in
$\Lam'/L$ acts by $\exp(2\pi i\eta(u,z))$ on $z\in\Lam$.
The common invariant sublattice of these characters is $L$. For a sector represented by
$t\in\Lam'/L$, take the shifted momenta $z+t$, with $z\in\Lam$, and
the action
\[
 g_u V_{z+t}=e^{2\pi i\eta(u,z+t)}V_{z+t}.
\]
If $u$ is changed by $l\in L$, its phase is unchanged: both
$\eta(l,z)$ and $\eta(l,t)$ are integers. The $t$-representative may
also be changed by $l\in L$: replacing $t$ by
$t+l$ and $z$ by $z-l$ leaves the momentum $z+t$ unchanged. If
$Nu\in L$, the symmetry $g_u$ has order dividing $N$ on every twisted
sector, by the same integrality argument. The actions therefore give
representations of $\Lam'/L$, and their phases multiply under addition
of momenta.

Since $\Lam'$ is integral, $\eta(u,t)\in\Z$. Projection in the
$t$-sector retains exactly those $z\in\Lam$ that pair integrally with
all of $\Lam'$, namely
$z\in\Lam\cap(\Lam')^*=L$. Thus the projected sector is $L+t$,
and the union of these sectors is $\Lam'$. Its evenness gives integral
spin, and its usual lattice cocycles give mutually local vertex
operators extending those on the common sublattice $L$.

Define the twisted traces, suppressing the $\bar\tau$ argument, by
\[
 Z_{t,u}(\tau)=\frac1{|\eta_D(\tau)|^{6d}}
 \sum_{p\in\Lam+t}e^{2\pi i\eta(u,p)}
 q^{|p_L|^2/2}\bar q^{|p_R|^2/2}.
\]
Evenness of $\Lam$ and Poisson summation on the unimodular lattice
$\Lam$, respectively, give
\begin{align*}
 Z_{t,u}(\tau+1)&=e^{-\pi i\eta(t,t)}Z_{t,u+t}(\tau)
                  =Z_{t,u+t}(\tau),\\
 Z_{t,u}(-1/\tau)&=e^{2\pi i\eta(t,u)}Z_{u,-t}(\tau)
                  =Z_{u,-t}(\tau).
\end{align*}
For the first identity write $p=z+t$ and use
$\eta(p,p)/2\equiv\eta(t,p)-\eta(t,t)/2\pmod\Z$.
For the second, the Fourier transform shifts the dual sum by $-u$;
its Gaussian factor $|\tau|^{3d}$ cancels that from the oscillator
denominator, as in Proposition~\ref{prop:modular}. Reflection of the
summation variable replaces $(-u,t)$ by $(u,-t)$. Both displayed
phases are one because $t,u\in\Lam'$ and $\Lam'$ is even integral.
Character orthogonality gives
\begin{equation}\label{eq:orbifoldsum}
 \frac1{|G|}\sum_{t,u\in\Lam'/L}Z_{t,u}
 =\frac1{|\eta_D(\tau)|^{6d}}
   \sum_{p\in\Lam'}q^{|p_L|^2/2}\bar q^{|p_R|^2/2}
 =Z_{\C'}.
\end{equation}
\end{proof}

An abstract group and its untwisted characters alone need not specify
an orbifold: the action in the twisted sectors, including any discrete
torsion, also matters. The lifts in $\Lam'$ in the proof fix this data.
Reversing the two codes gives the inverse relation through the quantum
symmetry of the orbifold.

\subsection{Momentum and winding phases}
In the convention \eqref{eq:narainform}, a shift vector $t=(r,s)$ acts as
\begin{equation}\label{eq:shiftcharges}
 \chi_t(k,\ell)=e^{2\pi i(\alpha^Tk+\delta^T\ell)},\qquad
 \alpha=\gamma^{-1}s,\quad \delta=\gamma^Tr-B\alpha.
\end{equation}
Regard $\alpha,\delta$ as elements of $\R^n/\Z^n$. For a nontrivial
character, the shift type is
\[
\begin{array}{c|c|c}
\text{type}&\text{condition}&\text{translation}\\\hline
\mathrm P&\delta=0&\text{torus}\\
\mathrm W&\alpha=0&\text{dual torus}\\
\mathrm M&\alpha\ne0,\ \delta\ne0&\text{mixed}
\end{array}
\]
These conditions depend on the null polarization.

Table~\ref{tab:orbifolds} lists the directed gaugings between the fixed
representatives of Table~\ref{tab:d2}. The words $v_j$ generate $\C'/I$
and act by \eqref{eq:finitecharacter}. Types are listed in generator
order, using the product frames of Proposition~\ref{prop:essentialB}
for the first six source classes and \eqref{eq:GM}--\eqref{eq:BM} for
$\mathrm M_2$. In a product frame, dependence on $a$ gives momentum
phases and dependence on $b$ gives winding phases on the source code.

\begin{table}[htbp]
\centering\small
\renewcommand{\arraystretch}{1.08}
\begin{tabular}{@{}ccclc@{}}
\toprule
Source & Target & $G$ & Generators $v_j$ & Types\\
\midrule
$\mathrm F_0$ & $\mathrm F_1$ & $\Z_4$ & $1300$ & W\\
$\mathrm F_0$ & $\mathrm F_2$ & $\Z_4$ & $2100$ & W\\
$\mathrm F_0$ & $\mathrm M_0$ & $\Z_2$ & $2000$ & W\\
$\mathrm F_0$ & $\mathrm M_1$ & $\Z_2$ & $2200$ & W\\
$\mathrm F_0$ & $\mathrm M_2$ & $\Z_4\times\Z_2$ & $1113,\ 0202$ & W,W\\
$\mathrm F_0$ & $\mathrm T$ & $\Z_2^2$ & $0200,\ 2000$ & W,W\\
$\mathrm F_1$ & $\mathrm F_2$ & $\Z_4^2$ & $0012,\ 2100$ & P,W\\
$\mathrm F_1$ & $\mathrm M_0$ & $\Z_4\times\Z_2$ & $0001,\ 2000$ & P,W\\
$\mathrm F_1$ & $\mathrm M_1$ & $\Z_2$ & $0002$ & P\\
$\mathrm F_1$ & $\mathrm M_2$ & $\Z_2$ & $0202$ & M\\
$\mathrm F_1$ & $\mathrm T$ & $\Z_2^2$ & $0002,\ 0200$ & P,W\\
$\mathrm F_2$ & $\mathrm M_0$ & $\Z_4\times\Z_2$ & $0001,\ 2000$ & P,W\\
$\mathrm F_2$ & $\mathrm M_1$ & $\Z_4\times\Z_2$ & $0011,\ 2200$ & P,W\\
$\mathrm F_2$ & $\mathrm M_2$ & $\Z_4\times\Z_2$ & $1113,\ 0022$ & M,P\\
$\mathrm F_2$ & $\mathrm T$ & $\Z_2^2$ & $0002,\ 2000$ & P,W\\
$\mathrm M_0$ & $\mathrm M_1$ & $\Z_2^2$ & $0011,\ 2200$ & P,W\\
$\mathrm M_0$ & $\mathrm M_2$ & $\Z_4$ & $1113$ & M\\
$\mathrm M_0$ & $\mathrm T$ & $\Z_2$ & $0200$ & W\\
$\mathrm M_1$ & $\mathrm M_2$ & $\Z_2^2$ & $0202,\ 1113$ & W,W\\
$\mathrm M_1$ & $\mathrm T$ & $\Z_2$ & $0200$ & W\\
$\mathrm M_2$ & $\mathrm T$ & $\Z_2$ & $0002$ & W\\

\bottomrule
\end{tabular}
\caption{Shift orbifolds between the fixed labeled representatives.
The target code fixes the twisted-sector projection.}
\label{tab:orbifolds}
\end{table}

\begin{samepage}
For example, $\mathrm F_1\to\mathrm M_2$ has generator $0202$ and
character $(-1)^{a_2+b_2}$, a mixed shift in the indicated product
frame. For $\mathrm F_0\to\mathrm M_2$, the source has $a=0$; the
generators $1113,0202$ act by $i^{3(b_1+b_2)}$ and $(-1)^{b_2}$.
Both are winding phases; the target lifts specify their twisted-sector
extensions.\par
\end{samepage}

For each row, put $I=\C\cap\C'$. The invariant factors follow from
\[
 |G|=\frac{16}{|I|},\qquad
 2G\simeq(2\C'+I)/I,\qquad
 \left.\begin{aligned}|2G|&=2^a\\|G|&=2^{2a+b}\end{aligned}\right\}
 \Longrightarrow G\simeq\Z_4^a\oplus\Z_2^b.
\]
The listed words generate $G$ with these orders modulo $I$; their
characters have common kernel $I$ on $\C$.

Independent block-monomial changes of $\C$ and $\C'$ can change $I$.
For example, with $\C'=\mathrm F_1$,
\[
\begin{array}{c|c|c}
\C&|\C\cap\C'|&G\\\hline
\mathrm F_0&4&\Z_4\\
\langle1000,0001\rangle\sim\mathrm F_0&1&\Z_4^2
\end{array}
\]
Thus $G$ depends on the relative placement of the codes, not only on
their separate block-monomial classes.

\begin{remark}[The choice of twisted-sector phases]\label{rem:torsionexample}
For $\mathrm M_1\to\mathrm M_2$, $G=\Z_2^2$ and
$I=\langle2200,0022\rangle$. With $g=(1,1,1)$, choose
\[
 r_1=s_1=(0,g),\qquad
 r_2=(g/2,g/2),\qquad s_2=(g/2,-g/2)
\]
in $W^2$. The target shifts $t_i=(r_i,s_i)$ have labels $0202,1113$.
Since $\lambda_i=(0,s_i)\in\Lam_{\mathrm M_1}$, the pure winding
shifts $\widetilde t_i=(r_i,0)$ give the same untwisted characters.
In the sector shifted by $\beta_1\widetilde t_1+\beta_2\widetilde t_2$,
replacing the insertion $\sum_i\alpha_i\widetilde t_i$ by
$\sum_i\alpha_i t_i$ multiplies its action by
\[
 \varepsilon(\alpha,\beta)
 =\exp\!\left(2\pi i\sum_{i,j}\alpha_i\beta_j s_i\cdot r_j\right)
 =(-1)^{\alpha_1\beta_2+\alpha_2\beta_1},
 \qquad \alpha,\beta\in\Z_2^2.
\]
This is the nontrivial discrete-torsion phase for $\Z_2^2$.
Without it, the pure shifts give the code
\[
 I+\langle0200,1100\rangle
 =\langle02,11\rangle\times\langle22\rangle,
\]
which is equivalent to $\mathrm M_1$ by exchanging the two null
coordinate groups. With this phase the target is $\mathrm M_2$, by
Proposition~\ref{prop:orbifold}. The two extensions of the same winding action therefore give
distinct target backgrounds.
\end{remark}

\subsection{The orbifold from \texorpdfstring{$\mathrm T$ to $\mathrm M_2$}{T to M2}}
For $\mathrm T$ the momentum lattice is $\Z^{12}$ with its standard
positive form. Equations~\eqref{eq:D12extension}--\eqref{eq:M2positive}
give
\[
 \Lam_{\mathrm T}\cap\Lam_{\mathrm M_2}=D_{12},\qquad
 \Lam_{\mathrm M_2}=D_{12}\cup(s+D_{12}).
\]
Here $2s\in D_{12}$, $\eta(s,s)=2$, and $s\notin\Z^{12}$.
The order-two symmetry on the $\mathrm T$ theory is
\begin{equation}\label{eq:Tparity}
 \chi_s(z)=e^{2\pi i\eta(s,z)}
 =(-1)^{\sum_jz_j}=(-1)^{\HH(z,z)}=(-1)^{2\Delta(z)}.
\end{equation}
These identities follow from the coordinates of $s$ modulo $\Z$ and
$z_j^2\equiv z_j\pmod2$. They extend to oscillator descendants since
$N+\bar N\in\Z$. Multiplicativity of $(-1)^{2\Delta}$ here follows
from its equality with the lattice character $\chi_s$.

The projected sectors are $D_{12}$ and $s+D_{12}$, whose union is
$\Lam_{\mathrm M_2}$. In the $(G,B)=(I_6,0)$ frame of $\mathrm T$,
\eqref{eq:Tparity} is a mixed momentum and winding shift. The inverse quantum symmetry is
the last row of Table~\ref{tab:orbifolds}: $v=0002$ acts on
$\mathrm M_2$ by $(-1)^{a_2}$. In its displayed frame it is
$(-1)^{\ell_1}$, a winding phase, as follows from
\eqref{eq:shiftcharges} with $(r,s)=(0,e_6)$,
$\gamma=2H_M^{-T}$ and $-B_Me_6=\tfrac12e_1$.

\end{document}